\documentclass[a4paper,UKenglish]{lipics}
\usepackage{tikz}
\usetikzlibrary{shapes}
\tikzstyle{every picture} = [>=latex]
\theoremstyle{plain}
\newtheorem{proposition}[theorem]{Proposition}
\newtheorem{conjecture}[theorem]{Conjecture}
\def\ca#1{{\cal#1}}
\def\crg{\mathop{\rm cr}}
\def\tcrn{\mathop{\rm tcr}}
\def\acrn{\mathop{\rm acr}}
\let\sem\setminus
\newcommand{\tuple}[1]{\langle{#1}\rangle}  

\def\nokernelhypo{\mbox{NP\,$\subseteq$\,coNP/poly}}

\title{Crossing Number is Hard for Kernelization}
\author{Petr Hlin\v en\'y and Marek Der\v n\'ar}
\affil{Faculty of Informatics, Masaryk University Brno, Czech Republic\\
	\texttt{hlineny@fi.muni.cz, m.dernar@gmail.com}}
\authorrunning{P.~Hlin\v en\'y and M.~Der\v n\'ar}

\Copyright{Petr Hlin\v en\'y and Marek Der\v n\'ar}


\begin{document}

\maketitle

\begin{abstract}
The graph crossing number problem, $\crg(G)\leq k$,
asks for a drawing of a graph $G$ in the plane with at most $k$ edge crossings.
Although this problem is in general notoriously difficult, 
it is fixed-parameter tractable for the parameter~$k$ [Grohe].
This suggests a closely related question of whether this problem
has a {\em polynomial kernel},
meaning whether every instance of $\crg(G)\leq k$ can be in polynomial time
reduced to an equivalent instance of size polynomial in~$k$
(and independent of~$|G|$).
We answer this question in the negative.
Along the proof we show that the tile crossing number problem of twisted
planar tiles is NP-hard, which has been an open problem for some time, too,
and then employ the complexity technique of cross-composition.
Our result holds already for the special case of graphs obtained 
from planar graphs by adding one edge.
%
\end{abstract}

\section{Introduction}
\label{sec:intro}

We refer to Sections~\ref{sec:crossingn},\ref{sec:parameterized}
for detailed formal definitions.
Briefly, the {\em crossing number $\crg(G)$} of a graph $G$ 
is the minimum number of pairwise edge crossings in a drawing of $G$ in the plane. 
Finding the crossing number of a graph is one of the most prominent hard
optimization problems in geometric graph theory \cite{GJ} and is NP-hard
already in very restricted cases, e.g., for cubic
graphs~\cite{DBLP:journals/jct/Hlineny06a}, and for graphs
with prescribed edge rotations~\cite{DBLP:journals/algorithmica/PelsmajerSS11}.
Concerning approximations, there exists $c>1$ such that the crossing number 
cannot be approximated within the factor $c$ in polynomial
time~\cite{DBLP:journals/dcg/Cabello13}.
Moreover, the following very special case of the problem is still hard -- a
result that greatly inspired our paper:
\begin{theorem}[Cabello and Mohar~\cite{DBLP:journals/siamcomp/CabelloM13}]
\label{thm:CM}
Let $G$ be an almost-planar graph, i.e., $G$ having an edge $e\in E(G)$ such
that $G\sem e$ is planar 
(called also {\em near-planar} in~\cite{DBLP:journals/siamcomp/CabelloM13}).
Let $k\geq1$ be an integer.
Then it is NP-complete to decide whether $\crg(G)\leq k$.
\end{theorem}

On the other hand, it has been shown that the problem is {\em fixed-parameter
tractable} when parameterized by itself: one can decide whether $\crg(G)\leq k$
in quadratic (Grohe~\cite{DBLP:conf/stoc/Grohe01}) 
and even linear (Kawarabayashi--Reed~\cite{DBLP:conf/stoc/KawarabayashiR07}) 
time while having $k$ fixed.
Fixed-parameter tractability (FPT) is closely related to the concept of
so called {\em kernelization}.
In fact, one can easily show that a (decidable) problem $\ca A$ parameterized 
by an integer $k$ is FPT if,
and only if, every instance of $\ca A$ can be in polynomial time reduced to
an equivalent instance (the {\em kernel}) of size bounded only by some function of~$k$.
This function of $k$, bounding the kernel size, may in general be arbitrarily huge.
Though, the really interesting case is when the kernel size may be bounded
by a polynomial function of~$k$ (a~{\em polynomial kernel}).

The nature of the methods used in \cite{DBLP:conf/stoc/Grohe01,DBLP:conf/stoc/KawarabayashiR07},
together with the recent great advances in algorithmic graph minors theory,
might suggest that the crossing number problem $\crg(G)\leq k$ should have a
polynomial kernel in~$k$, as many related FPT problems do.
This question was raised as open, e.g., at WorKer 2015 [unpublished].
Polynomial kernels for some special crossing number problem instances
were constructed before, e.g., in~\cite{DBLP:conf/gd/BannisterES13}.
The general result is, however, very unlikely to hold as our main result claims:

\begin{theorem}
\label{thm:nokernel}
Let $G$ be an {\em almost-planar} graph, i.e., $G$ having an edge $e\in E(G)$ such
that $G\sem e$ is planar.
Let $k\geq1$ be an integer.
The crossing number problem, asking if $\crg(G)\leq k$ while parameterized by~$k$,
does not admit a polynomial kernel unless \nokernelhypo.
\end{theorem}

In order to prove Theorem~\ref{thm:nokernel}, we use the technique of
{\em cross-composition}~\cite{DBLP:journals/siamdm/BodlaenderJK14}.
While its formal description is postponed till Section~\ref{sec:parameterized}, 
here we very informally outline the underlying idea of cross-composition.
Imagine we have an NP-hard language $\ca L$ such that we can
``{\sc or}-cross-compose'' an arbitrary collection of instances $x_1,x_2,\dots,x_t$
of $\ca L$ into the crossing number problem $\crg(G_0)\leq k_0$
for suitable $G_0$ and $k_0$ efficiently depending on $x_1,x_2,\dots,x_t$.
By the words ``{\sc or}-cross-compose'' we mean that $\crg(G_0)\leq k_0$
holds if and only if $x_i\in\ca L$ for some $1\leq i\leq t$
(informally, $x_1\in\ca L$ {\sc or} $x_2\in\ca L$ {\sc or} \dots).
Now assume we could always reduce a crossing number instance $\tuple{G,k}$ into
an equivalent instance of size $p(k)$ where $p$ is a polynomial.
Then, for the instance $\tuple{G_0,k_0}$ and suitable $t$ such that 
$p(k_0)<\!<t\approx|G_0|$\mbox{$\,<\!<2^{|x_i|}$},
such a reduction effectively means that we should somehow decide many of the
$t$ instances $x_i\in\ca L$ in time polynomial in $|G_0|$ (which is $<\!<2^{|x_i|}$).
The latter sounds highly unlikely \cite{DBLP:journals/jcss/FortnowS11} 
in the complexity theory.

The task is to find a suitable NP-hard language $\ca L$ for the
aforementioned construction.
While the ordinary crossing number problem is not suitable for cross-composition
(roughly, since the crossing numbers of disjoint instances sum up together),
a helping hand is given by the concept of the {\em tile crossing number}
\cite{DBLP:journals/jgt/PinontoanR03}, defined in detail in Section~\ref{sec:crossingn}.

Informally, a {\em tile} is a graph $T$ with two disjoint sequences of vertices 
defining the {left and right walls} of~$T$.
A {\em tile drawing} is a drawing of $T$ inside a rectangle such that the
walls of $T$ lie respectively on the left and right sides of this rectangle.
A tile $T$ is planar if $T$ admits a tile drawing without crossings,
and $T$ is {\em twisted planar} if $T$ becomes a planar tile after inverting
(upside-down) one of the walls.
As observed by Schaefer~\cite{survscha}, the tile crossing number problem is NP-hard
by a trivial reduction from ordinary crossing number, but we need much more.
In order to embed the tile crossing number problem in a cross-composition construction,
which will be realized as a concatenation of the tile instances across their respective
walls, we shall use only twisted planar tiles.
See Figure~\ref{fig:tilecomposition}.
The underlying idea which makes the cross-composition work,
is that only one of the tile instances is drawn
twisted in the concatenation and all the other contribute no crossings.

\begin{figure}[ht]
\begin{center}\bigskip
\begin{tikzpicture}[scale=0.5]
\normalsize
\tikzstyle{every node}=[draw, shape=circle, minimum size=2.5pt,inner sep=1.5pt, fill=black]
\tikzstyle{every path}=[color=lightgray, fill=lightgray]
\draw (0,0) -- (4,2) -- (4,0) -- (0,2) -- (0,0) ;
\draw (4,0) -- (8,2) -- (8,0) -- (4,2) -- (4,0) ;
\draw (8,0) -- (12,2) -- (12,0) -- (8,2) -- (8,0) ;
\draw (12,0) -- (16,2) -- (16,0) -- (12,2) -- (12,0) ;
\draw (16,0) -- (20,2) -- (20,0) -- (16,2) -- (16,0) ;
\tikzstyle{every path}=[color=black]
\draw (0,0) node (x1) {}; \draw (0,2) node[fill=red] (x2) {};
\draw[line width=1.2pt, dashed] (x2) -- (x1);
\draw (4,2) node[fill=red] (x3) {}; \draw (4,0) node (x4) {};
\draw (x1) -- (x3) -- (x4) -- (x2) -- (x1);
\draw (8,0) node (x5) {}; \draw (8,2) node[fill=red] (x6) {};
\draw (x3) -- (x5) -- (x6) -- (x4);
\draw (12,2) node[fill=red] (x7) {}; \draw (12,0) node (x8) {};
\draw (x5) -- (x7) -- (x8) -- (x6);
\draw (16,0) node (x9) {}; \draw (16,2) node[fill=red] (x10) {};
\draw (x7) -- (x9) -- (x10) -- (x8);
\draw (20,2) node[fill=red] (x11) {}; \draw (20,0) node (x12) {};
\draw (x9) -- (x11) -- (x12) -- (x10);
\draw[line width=1.2pt, dashed] (x12) -- (x11);
\end{tikzpicture}
\\{\boldmath$\downarrow$}\\[3ex]
\begin{tikzpicture}[scale=0.5]
\normalsize
\tikzstyle{every node}=[draw, shape=circle, minimum size=2.5pt,inner sep=1.5pt, fill=black]
\tikzstyle{every path}=[color=lightgray, fill=lightgray]
\draw (x1) rectangle (16,2) ;
\draw (16,0) -- (20,2) -- (20,0) -- (16,2) -- (16,0) ;
\tikzstyle{every path}=[color=black]
\draw (0,0) node (x1) {}; \draw (0,2) node[fill=red] (x2) {};
\draw[line width=1.2pt, dashed] (x2) -- (x1);
\draw (4,0) node[fill=red] (x3) {}; \draw (4,2) node (x4) {};
\draw (x1) -- (x3) -- (x4) -- (x2) -- (x1) ;
\draw (8,0) node (x5) {}; \draw (8,2) node[fill=red] (x6) {};
\draw (x3) -- (x5) -- (x6) -- (x4);
\draw (12,0) node[fill=red] (x7) {}; \draw (12,2) node (x8) {};
\draw (x5) -- (x7) -- (x8) -- (x6);
\draw (16,0) node (x9) {}; \draw (16,2) node[fill=red] (x10) {};
\draw (x7) -- (x9) -- (x10) -- (x8);
\draw (20,2) node[fill=red] (x11) {}; \draw (20,0) node (x12) {};
\draw (x9) -- (x11) -- (x12) -- (x10);
\draw[line width=1.2pt, dashed] (x12) -- (x11);
\end{tikzpicture}
\end{center}
\caption{Schematic concatenation of an odd number of twisted planar tiles;
	in fact, only one (and an arbitrary one) 
	of the tiles needs to be drawn twisted in this case.}
\label{fig:tilecomposition}
\end{figure}
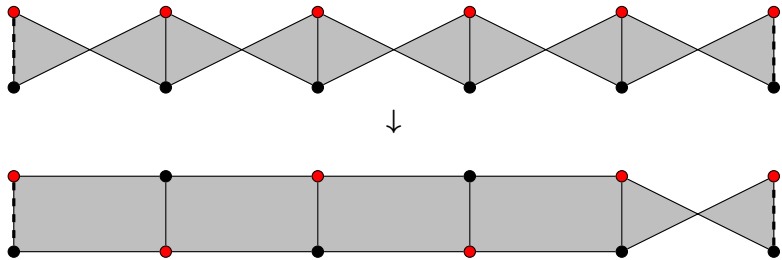

Hence the proof of Theorem~\ref{thm:nokernel} would be finished,
modulo technical details, if we show
that the tile crossing number problem of twisted planar tiles is NP-hard.
This particular question seems to have been latently considered in the
crossing number community for the past several years, and it is still open
nowadays to our best knowledge.
We provide the following affirmative answer by adapting
a construction from the proof
\cite{DBLP:journals/siamcomp/CabelloM13} of Theorem~\ref{thm:CM}:

\begin{theorem}[Corollary~\ref{cor:diagseptile}]
\label{thm:twistedhard}
Let $T$ be a twisted planar tile and $k\geq1$ an integer.
Then it is NP-complete to decide whether there exists a tile drawing 
of $T$ with at most $k$ edge crossings.
Furthermore, the same holds if both the walls of $T$ are of size two and
there exists an edge $e\in E(T)$ such that $T\sem e$ is a planar tile.
\end{theorem}

\subparagraph*{Paper organization. }
We provide the necessary formal definitions of the aforementioned concepts
from crossing numbers and parameterized complexity in
Sections~\ref{sec:crossingn},\ref{sec:parameterized}.
Then we prove Theorem~\ref{thm:twistedhard} in Section~\ref{sec:twistedh},
and provide technical claims useful for the next cross-composition
construction in Section~\ref{sec:cross-composing}.
Finally, we summarize the paper and present some additional ideas in
Section~\ref{sec:conclusion}.

\section{Crossing numbers}
\label{sec:crossingn}

We consider multigraphs by default, even though we could always subdivide
parallel edges in order to make the graphs simple.
We follow basic terminology of topological graph theory,
see e.g.~\cite{MT}.
A {\em drawing} of a graph $G$ in the plane is such that, 
the vertices of $G$ are distinct points 
and the edges are simple curves joining their endvertices.
It is required that no edge passes through a vertex, 
and no three edges cross in a common point.
\begin{definition}[crossing number]
The {\em crossing number} $\crg(G)$ of a graph $G$
is the minimum number of crossing points of edges in a drawing of $G$ in the plane.
\end{definition}
Hence, a graph $G$ is planar if and only if~$\crg(G)=0$.
Note that the crossing number is invariant under subdividing edges of~$G$.

A useful concept in crossing numbers research are tiles. 
They were used already
by Kochol \cite{DBLP:journals/dm/Kochol87} and Richter--Thomassen
\cite{DBLP:journals/jct/RichterT93},
although they were formalized only later in the work of Pinnontoan and Richter 
\cite{DBLP:journals/jgt/PinontoanR03,Pinon:IOPORT.02108147}. 
So far, primary use of the tile concept in crossing numbers research concerned
study of so called crossing-critical graphs,
as can be seen also in recent papers such as~\cite{DBLP:conf/gd/BokalBDH15,cit:2critchar}.
Here we will use tiles in a rather different way.
We briefly sketch the necessary terms as follows.

A {\em tile} is a triple $T=(G,\lambda, \rho)$
where $\lambda,\rho\in V(G)^*$ are two disjoint sequences of distinct
vertices of~$G$, called the {\em left and right wall} of~$T$, respectively.
A {\em tile drawing} of $T$ is a drawing of the underlying graph $G$ in the
unit square such that the vertices of $\lambda$ occur in this order on the
left side of the square and those of $\rho$ in this order on the right side
of it.
The {\em tile crossing number $\tcrn(T)$} of a tile $T$ is the
minimum number of crossing points of edges over all tile drawings of $T$.
The {\em right-inverted} tile $T^{\updownarrow}$ is the tile 
$(G, \lambda, \bar{\rho})$ and the {\em left-inverted} 
tile $^{\updownarrow} T$ is $(G, \bar{\lambda}, \rho)$,
where $\bar{\lambda}$ and $\bar{\rho}$ denote the inverted sequences
of~$\lambda,\rho$. 

For simplicity, in this brief exposition, we shall assume that all 
tiles involved in one construction satisfy 
$|\lambda|=|\rho|=w$ for suitable $w\geq2$
(though, a more general treatment is obviously possible).
The {\em join of two tiles} $T=(G,\lambda, \rho)$ and $T'=(G', \lambda', \rho')$ 
is defined as the tile $T \otimes T':=(G'', \lambda, \rho')$, where 
$G''$ is the graph obtained from the disjoint union of $G$ and $G'$,
by identifying $\rho(i)$ with $\lambda'(i)$ for $i=1,\ldots,w$. 
Since the operation $\otimes$ is associative, we can safely define the
{join of a sequence of tiles} $\ca T=(T_1,T_2, \ldots, T_m)$ as
the tile given by 
$\otimes \ca T=T_1 \otimes T_2 \otimes \ldots \otimes T_m$.

A tile $T=(G, \lambda, \rho)$ is {\em planar} if $\tcrn(T)=0$,
and $T$ is {\em twisted planar} if $\tcrn(T^{\updownarrow})=0$
(which is clearly equivalent to $\tcrn({}^{\updownarrow}T)=0$\,).
We briefly illustrate these definitions (also
Figure~\ref{fig:tilecomposition}):
\begin{example}
\label{ex:tilesjoincr}
Let $\ca T=(T_1,T_2, \ldots, T_m)$ be a sequence of twisted planar tiles
$T_i$, $i=1,\dots,m$.
Then $\tcrn(\otimes \ca T)=0$ if $m$ is even, and
$\tcrn(\otimes \ca T)\leq\min_{i\in\{1,\dots,m\}}\tcrn(T_i)$ otherwise.
\end{example}

Finally, the following is a useful artifice 
in crossing numbers research. In a {\em weighted} graph, each edge is
assigned a positive number (the {\em weight, or thickness} of the edge). 
Now the {crossing number} is defined as in the ordinary case,
but a crossing point between edges $e_1$ and $e_2$,
say of weights $t_1$ and $t_2$, contributes $t_1\cdot t_2$ to the result.
In the case of integer weights, this extension can be easily seen equivalent
to the unweighted setting as follows:
\begin{proposition}[folklore]
\label{pro:weighted}
Let $G$ be an integer-weighted graph, $F\subseteq E(G)$, 
and $G^+$ be constructed from $G$ via replacing each edge $e\in F$ 
of weight $t$ with a bunch of $t$ parallel edges of weight~$1$.
Then $\crg(G)=\crg(G^+)$.
Moreover, if $G$ is the graph of a tile $T$ and $T^+$ is the corresponding
tile based on $G^+$, then $\tcrn(T)=\tcrn(T^+)$.
\end{proposition}

\section{Parameterized complexity and kernelization}
\label{sec:parameterized}

Here we introduce the relevant concepts of parameterized complexity theory.
For more details, we refer to textbooks~\cite{DBLP:series/txcs/DowneyF13,fg06}.
Let $\Sigma$ be a finite alphabet.
A~parameterized problem over $\Sigma$ is a language 
$\ca A\subseteq \Sigma^*\times\mathbb N$.
An instance of $\ca A$ is thus a pair $\tuple{x,k}$ where $x$ is
the input and $k\geq0$ an (integer) parameter.  
In our case, e.g., $\tuple{G,k}$ is the crossing number instance
``$\crg(G)\leq k$''.
A~parameterized problem is \emph{fixed-parameter tractable} (FPT) if every instance
$\tuple{x,k}$ can be solved in time $f(k)\cdot|x|^c$, where $f$ is a
computable function and $c$ is a constant.

A hot research direction in the area of parameterized complexity of the past
decade is that of kernelization.
A \emph{kernelization}
for a parameterized problem~$\ca A$ is an algorithm that takes an instance
$\tuple{x,k}$ of $\ca A$ and, in time polynomial in $|x|+k$, 
maps $\tuple{x,k}$ to an equivalent 
instance $\tuple{x',k'}$ of $\ca A$ such that $|x'|+k'\leq f(k)$
where $f$ is a computable function.
The output $\tuple{x',k'}$ is called the \emph{kernel}. 
We say that $\ca A$ has a \emph{polynomial kernel} if
there is a kernelization for $\ca A$ such that $f$ is a polynomial.  
Every fixed-parameter tractable problem admits a kernel, 
but not necessarily a polynomial kernel. 

We now describe the basic {\sc or}-cross-composition framework of
\cite{DBLP:journals/siamdm/BodlaenderJK14}.
An equivalence relation $\sim$ on $\Sigma^*$ is called a polynomial
equivalence if, for any $x,y\in\Sigma^*$, we can decide in polynomial time
whether $x\sim y$ and, moreover,
on any finite $S\subseteq\Sigma^*$ the relation $\sim$ defines a number of
equivalence classes which is polynomially bounded in the size of a largest
element of~$S$.
For our purpose, $\sim$ will group together the tile crossing number
instances of the same objective value~$k$.
\begin{definition}[{\sc or}-cross-composition]
\label{def:crosscomposition}
Let $\ca L\subseteq\Sigma^*$ be a language, 
$\sim$ be a polynomial equivalence relation on $\Sigma^*$,
and let $\ca A\subseteq \Sigma^*\times\mathbb N$ be a parameterized problem.
An {\em{\sc or}-cross-composition} of $\ca L$ into $\ca A$
is an algorithm that, given $t$ instances
$x_1,x_2,\dots,x_t\in\Sigma^*$ of $\ca L$ such that
$x_1\sim x_2\sim\dots\sim x_t$;
\begin{itemize}
\item in time polynomial in $|x_1|+\dots|x_t|$ it outputs an instance
$\tuple{y_0,k_0}\in \Sigma^*\times\mathbb N$ such that
$k_0$ is polynomially bounded in $\max_i|x_i|+\log t$, and
\item $\tuple{y_0,k_0}\in\ca A$ if and only if $x_i\in\ca L$ for some 
$1\leq i\leq t$.
\end{itemize}
\end{definition}

\begin{theorem}[Bodlaender, Jansen and Kratsch~\cite{DBLP:journals/siamdm/BodlaenderJK14}]
\label{thm:nopolykernel}
If an NP-hard language $\ca L$ has an
{\sc or}-cross-composition into the parameterized problem $\ca A$,
then $\ca A$ does not admit a polynomial kernel unless \nokernelhypo.
\end{theorem}
We remark in passing that the full claim of
\cite{DBLP:journals/siamdm/BodlaenderJK14} is even stronger than stated
Theorem~\ref{thm:nopolykernel}, and in particular it also excludes the
existence of a so-called polynomial compression of~$\ca A$.

\section{Twisted planar tiles}
\label{sec:twistedh}

For the purpose of our proof,
we are especially interested in the following kind of integer-weighted planar tiles.
See Figure~\ref{fig:diagonalsep} for an illustration.

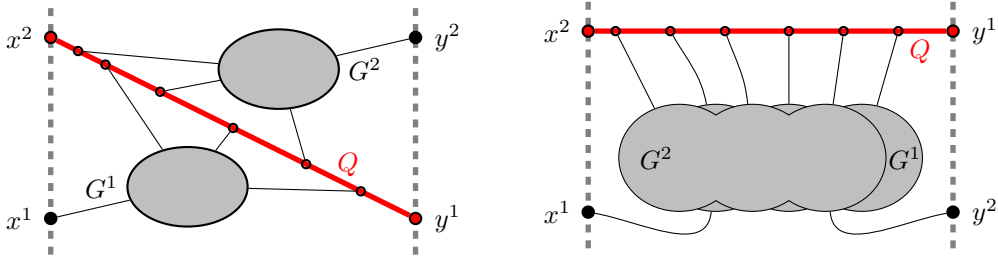
\begin{figure}[t]
\begin{center}\bigskip
\begin{tikzpicture}[scale=1.2]
\normalsize
\tikzstyle{every node}=[draw, thick, shape=circle, minimum size=3pt,inner sep=1.5pt, fill=black]
\tikzstyle{every path}=[color=black]
\draw[line width=2pt, dashed, color=gray] (0,-0.4) -- (0,2.4);
\draw (0,0) node[label=left:$x^1$] (x1) {};
\draw (0,2) node[label=left:$x^2$, fill=red] (x2) {};
\draw[line width=2pt, dashed, color=gray] (4,-0.4) -- (4,2.4);
\draw (4,2) node[label=right:$y^2$] (y2) {};
\draw (4,0) node[label=right:$y^1$, fill=red] (y1) {};
\draw[line width=2pt, color=red] (x2) -- (y1);
\tikzstyle{every node}=[draw, thick, shape=circle, minimum size=2pt,inner sep=1.1pt, fill=none]
\draw (1.5,0.35) node[ellipse, minimum width=45pt, minimum height=30pt,
		fill=lightgray, label=left:$G^1$] (G1) {};
\draw (2.5,1.65) node[ellipse, minimum width=45pt, minimum height=30pt,
		fill=lightgray, label=right:$G^2$] (G2) {};
\draw (G1) -- (x1) ; \draw (G2) -- (y2) ;
\draw (2,1) node (d1) {}; \draw (0.3,1.85) node (d6) {};
\draw (2.8,0.6) node[label=right:{\color{red}~~$Q$}] (d2) {};
\draw (3.4,0.3) node (d3) {};
\draw (1.2,1.4) node (d4) {}; \draw (0.6,1.7) node (d5) {};
\draw (d1) -- (G1) -- (d3) ; \draw (d5) -- (G1) ;
\draw (d2) -- (G2) -- (d4) ; \draw (d6) -- (G2) ;
\end{tikzpicture}
\qquad
\begin{tikzpicture}[scale=1.2]
\normalsize
\tikzstyle{every node}=[draw, thick, shape=circle, minimum size=3pt,inner sep=1.5pt, fill=black]
\tikzstyle{every path}=[color=black]
\draw[line width=2pt, dashed, color=gray] (0,-0.4) -- (0,2.4);
\draw (0,0) node[label=left:$x^1$] (x1) {};
\draw (0,2) node[label=left:$x^2$, fill=red] (x2) {};
\draw[line width=2pt, dashed, color=gray] (4,-0.4) -- (4,2.4);
\draw (4,0) node[label=right:$y^2$] (y2) {};
\draw (4,2) node[label=right:$y^1$, fill=red] (y1) {};
\draw[line width=2pt, color=red] (x2) -- (y1);
\tikzstyle{every node}=[draw, thick, ellipse, minimum width=45pt,
			minimum height=40pt, fill=lightgray]
\draw (1.4,0.6) node (G1) {};\draw (2.2,0.6) node (G1) {};\draw (3.0,0.6) node (G1) {};
\draw (1.4,0.6) node[draw=none] (G1a) {};\draw (2.2,0.6) node[draw=none] (G1b) {};
\draw (3.0,0.6) node[draw=none] (G1c) {$G^1$\hspace*{-7.5ex}};
\draw (1.0,0.6) node (G2) {};\draw (1.8,0.6) node (G2) {};\draw (2.6,0.6) node (G2) {};
\draw (1.0,0.6) node[draw=none] (G2a) {\hspace*{-4ex}$G^2$};
\draw (1.8,0.6) node[draw=none] (G2b) {};\draw (2.6,0.6) node[draw=none] (G2c) {};
\tikzstyle{every node}=[draw, thick, shape=circle, minimum size=2pt,inner sep=1.1pt, fill=none]
\draw (G1a) to [in=350,out=265] (x1) ;
\draw (G2c) to [in=190,out=275] (y2) ;
\draw (2.2,2) node (d1) {}; \draw (0.3,2) node (d6) {}; \draw (2.8,2) node (d2) {};
\draw (3.4,2) node[label=below:{\color{red}~$Q$\hspace*{-3ex}}] (d3) {};
\draw (1.5,2) node (d4) {}; \draw (0.9,2) node (d5) {};
\draw (d1) -- (G1b) ; \draw (d3) -- (G1c) ;
\draw (d5) to [in=100,out=300] (G1a) ;
\draw (d2) -- (G2c) ; \draw (d6) -- (G2a) ;
\draw (d4) to [in=95,out=290] (G2b) ; 
\end{tikzpicture}
\end{center}
\caption{A diagonally separated tile and a possible drawing of 
	the corresponding right-inverted tile.
	The underlying graph of this tile contains two vertex-disjoint 
	subgraphs $G^1,G^2$ such that $V(G^1)\cup V(G^2)=V(G)\sem\{x^2,y^1\}$,
	and their drawings ``overlay'' each other on the right.}
\label{fig:diagonalsep}
\end{figure}

\begin{definition}[diagonally separated tile]
\label{def:diagonalsep}
Consider an integer-weighted planar tile $T=(G,\lambda, \rho)$ where
the walls are $\lambda=(x^1,x^2)$ and $\rho=(y^1,y^2)$ for some distinct
$x^1,x^2,y^1,y^2\in V(G)$.
We say that $T$ is {\em diagonally separated}
if we can write $G=G^1\cup G^2\cup Q$ such that
\begin{itemize}
\item $G^1,G^2$ are vertex-disjoint subgraphs of $G$ such that
$V(G^1)\cup V(G^2)=V(G)\sem\{x^2,y^1\}$,
\item $E(Q)=E(G)\sem(E(G^1)\cup E(G^2))$, ~$x^1,y^2\not\in V(Q)$,
and $Q$ is a ``thick'' path from $x^2$ to $y^1$ having each edge
of weight $t\geq w_1\cdot w_2+1$ where $w_i$ is the sum of weights of all
the edges of the subgraph $G^i\sem V(Q)$,
\item  $G$ is connected and both 
$G^1\sem V(Q)$ and $G^2\sem V(Q)$ are connected subgraphs,
and no edge of $G^1\cup G^2$ has both ends in~$V(Q)\cup\{x^1,y^2\}$,
\item
$x^1\in V(G^1)\sem V(Q)$ and $y^2\in V(G^2)\sem V(Q)$,
both the vertices $x^1,y^2$ are of degree one in~$G$
and the two incident edges have weight~$1$.
\end{itemize}
\end{definition}

Twisted diagonally separated planar tiles have the suitable ``or-composability''
property:
\begin{lemma}
\label{lem:twistedorcompose}
Let $\ca T=(T_1^{\updownarrow},T_2^{\updownarrow}, \ldots, T_m^{\updownarrow})$ 
be a sequence of tiles such that, for $i=1,\dots,m$, $T_i$~is 
a diagonally separated planar tile.
Let $U:=\otimes\ca T$ if $m$ is odd, 
and $U:=(\otimes\ca T)^{\updownarrow}$ otherwise.
Then 
$\tcrn(U) = \min_{i\in\{1,\dots,m\}} \tcrn(T_i^{\updownarrow}) \,.$
\end{lemma}

\begin{proof}
Let the underlying graph of $T_i$ be $G^1_i\cup G^2_i\cup Q_i$, as
anticipated by Definition~\ref{def:diagonalsep},
and let $t_i$ be the weight of~$E(Q_i)$.
Let $(x^1_i,x^2_i)$, $(y^2_i,y^1_i)$ be the left and right walls,
respectively, of $T_i^{\updownarrow}$.
By the definition of join $\otimes$, $y^1_i=x^2_{i+1}$ and hence
$Q:=Q_1\cup\dots\cup Q_m$ is a path from $x^2_1$ to $y^1_m$.
Similarly, $y^2_i=x^1_{i+1}$, and so $H_i:=G^2_{i}\cup G^1_{i+1}$ is a connected
component of $U\sem V(Q)$ for $i=1,\dots,m-1$.
See Figure~\ref{fig:diagonalsepU}.
For simplicity, we let $H_0:=G^1_1$ and $H_m:=G^2_m$ which are also
components of $U\sem V(Q)$.

\begin{figure}[t]
\begin{center}
\def\tileI#1#2#3#4{%
\begin{tikzpicture}[#1]
\tikzstyle{every node}=[draw, thick, shape=circle, minimum size=2pt,inner sep=1pt, fill=black]
\tikzstyle{every path}=[color=black]
\draw[line width=1.2pt, dashed, color=white] (4,-0.4) -- (4,2.4);
#2
\draw (0,0) node (x1) {};
\draw (0,2) node[fill=red] (x2) {};
\draw (4,2) node (y2) {};
\draw (4,0) node[fill=red] (y1) {};
\draw[line width=1.2pt, color=red] (x2) -- (y1);
\tikzstyle{every node}=[draw, thick, shape=circle, minimum size=2pt,inner sep=1pt, fill=none]
\draw (1.3,0.2) node[ellipse, minimum width=15pt, minimum height=10pt] (G1) {#3};
\draw (2.7,1.8) node[ellipse, minimum width=15pt, minimum height=10pt] (G2) {#4};
\draw (G1) -- (x1) ; \draw (G2) -- (y2) ;
\draw (2,1) node (d1) {}; \draw (0.3,1.85) node (d6) {};
\draw (2.8,0.6) node (d2) {};
\draw (3.4,0.3) node (d3) {};
\draw (1.2,1.4) node (d4) {}; \draw (0.6,1.7) node (d5) {};
\draw (d1) -- (G1) -- (d3) ; \draw (d5) -- (G1) ;
\draw (d2) -- (G2) -- (d4) ; \draw (d6) -- (G2) ;
\end{tikzpicture}
}\small
\tileI{scale=0.6}{\draw[line width=2pt, dashed, color=gray]
	(0,-0.4) -- (0,2.4);
	\draw (0,0) node[label=left:$x_1^1$] {};
	\draw (0,2) node[label=left:$x_1^2$] {};
	}{$G_1^1$}{$G_1^2$}\hspace*{-7pt}%
\tileI{xscale=-0.6,yscale=0.6}{}{$G_2^2$}{$G_2^1$}\hspace*{-7pt}%
\tileI{scale=0.6}{}{$G_3^1$}{$G_3^2$}\hspace*{-7pt}%
\tileI{xscale=-0.6,yscale=0.6}{}{$G_4^2$}{$G_4^1$}\hspace*{-7pt}%
\tileI{scale=0.6}{\draw[line width=2pt, dashed, color=gray] 
	(4,-0.4) -- (4,2.4);
	\draw (4,2) node[label=right:$y_5^2$] {};
	\draw (4,0) node[label=right:$y_5^1$] {};
	}{$G_5^1$}{$G_5^2$}\hspace*{-7pt}%
\end{center}
\caption{The planar tile $U^{\updownarrow}$ where $U$ for $m=5$ is 
	from the statement of Lemma~\ref{lem:twistedorcompose}.}
\label{fig:diagonalsepU}
\end{figure}
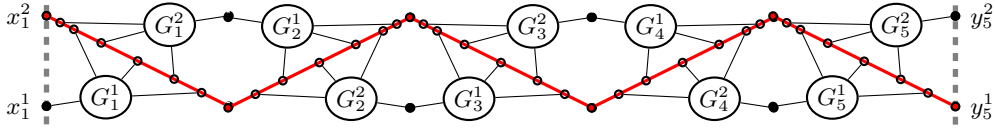%

It clearly holds 
$\tcrn(U) \leq \min_{i\in\{1,\dots,m\}} \tcrn(T_i^{\updownarrow})$.
Furthermore, we claim that $t_i>\tcrn(T_i^{\updownarrow})$ and so 
$t_i>\tcrn(U)$ for each $i=1,\dots,m$.
Since $T_i$ is planar, each of $G^1_i,G^2_i$ has a plane embedding
in which the vertices adjacent to $V(Q)\cup\{x_i^1,y_i^2\}$ lie on the outer face.
Consequently, there is a tile drawing of $T_i^{\updownarrow}$ with
each of $G^1_i,G^2_i$ plane and crossings only between
the edges of $G^1_i$ and of $G^2_i$ that are not incident to~$Q$.
See Figure~\ref{fig:diagonalsep} right.
By standard arguments, we may assume that no two edges cross more than once
in this drawing and so 
$\tcrn(T_i^{\updownarrow})\leq w_i^1\cdot w_i^2\leq t_i-1$
where $w_i^j$ is the sum of weights of all
the edges of~$G_i^j\sem V(Q)$.

From $t_i>\tcrn(U)$ for $i=1,\dots,m$ we get that
no edge of $Q$ is ever crossed in an optimal tile drawing of~$U$.
We may hence properly define, in any optimal tile drawing of $U$ and for each
subgraph~$H_i$, whether whole $H_i$ lies (is drawn) {\em above} or {\em below}~$Q$.
We aim to show that there always exists $i\in\{1,\dots,m\}$ such that
$H_{i-1},H_i$ are drawn on the same side of $Q$, either both above or both
below~$Q$.

Assume the contrary.
Then $H_0$ is drawn below $Q$ by the left wall $(x^1_1,x^2_1)$ of $U$.
Next, $H_1$ is drawn above, $H_2$ below, \dots, and finally, 
$H_m$ should be drawn above $U$ if $m$ is odd and below $U$ otherwise.
That is exactly the opposite position to what is requested by the right wall
of~$U$ which is $(y^2_m,y^1_m)$ if $m$ is odd
and $(y^1_m,y^2_m)$ otherwise, a contradiction.

So, $H_{i-1}$ and $H_i$ are drawn on the same side of $Q$ for some
$i\in\{1,\dots,m\}$.
First assume that $1<i<m$.
By supposed connectivity of $G_{i-1}$ there is a path 
$P^1\subseteq G^2_{i-1}$ from $x^1_i=y^2_{i-1}$ to an internal vertex of
$Q_{i-1}$, and similarly,
there is a path $P^2\subseteq G^1_{i+1}$ 
from $y^2_{i}=x^1_{i+1}$ to an internal vertex of $Q_{i+1}$
by connectivity of $G_{i+1}$.
Let $D$ be the drawing obtained from a considered optimal tile drawing of
$U$ restricted to~$T_i$,
by prolonging the single weight-$1$ edge incident with $x^1_i$ along $P^1$
and the single edge incident with $y^2_i$ along $P^2$.
Since whole $Q$ is uncrossed, the paths $Q_{i-1},Q_{i+1}$ can play the role
of the left and right wall of~$D$, and hence $D$ is a
valid tile drawing of~$T_i$ having no more crossings than~$\tcrn(U)$.

If $i=1$ or $i=m$, then we directly use the left or the right wall of $U$ in the
previous argument.
Consequently, $\min_{i\in\{1,\dots,m\}} \tcrn(T_i^{\updownarrow}) \leq \tcrn(U)$
and the proof is finished.
\end{proof}

The last step of this section is to prove that the tile crossing number problem
is NP-hard for twisted diagonally separated planar tiles.
Due to their similarity to intermediate steps in the paper
\cite{DBLP:journals/siamcomp/CabelloM13}, it is no surprise that we can
easily derive hardness using the same means;
from NP-hardness of the so called anchored crossing number.

An {\em anchored graph}~\cite{DBLP:journals/siamcomp/CabelloM13} 
is a triple $(G,A,\sigma)$, where $G$ is a graph,
$A\subseteq V(G)$ are the {anchor} vertices and $\sigma$ is a cyclic
ordering (sequence) of~$A$.
An {\em anchored drawing} of $(G,A,\sigma)$ is a drawing of $G$ in
a closed disc~$\Delta$ such that the vertices of $A$ are placed on the boundary
of $\Delta$ in the order specified by $\sigma$,
and the rest of the drawing lies in the interior of~$D$.
The {\em anchored crossing number $\acrn(G,A,\sigma)$},
or shortly $\acrn(G)$, is the minimum number of 
pairwise edge crossings in an anchored drawing of $(G,A,\sigma)$.
A {\em planar anchored graph} is
an anchored graph that has an anchored drawing without crossings.
Any subgraph $H\subseteq G$ naturally defines the corresponding
anchored subgraph $\big(H,A\cap V(H),\,\sigma\!\restriction\!V(H)\big)$. 

\begin{theorem}[Cabello and Mohar~\cite{DBLP:journals/siamcomp/CabelloM13}]%
\footnote{Note that \cite{DBLP:journals/siamcomp/CabelloM13} in general deals with
weighted crossing number, in the same way as we do e.g.\ in Proposition~\ref{pro:weighted}.
However, since their weights are always polynomial in the graph size,
Theorem~\ref{thm:anchoredhard} holds also for unweighted graphs.}
\label{thm:anchoredhard}
Let $G$ be an anchored graph that can be
decomposed into two vertex-disjoint connected planar anchored subgraphs.
Let $k\geq1$ be an integer.
Then it is NP-complete to decide whether $\acrn(G)\leq k$.
\end{theorem}

\begin{corollary}
\label{cor:diagseptile}
Let $T$ be a diagonally separated planar tile,
and $k\geq1$ be an integer.
Then it is NP-complete to decide whether $\tcrn(T^{\updownarrow})\leq k$.
\end{corollary}
Note that twisted diagonally separated planar tiles satisfy all the
assumptions of Theorem~\ref{thm:twistedhard}.
In particular, if $f$ denotes the edge incident to $x^1$ in $T$ then
both $T$ and $T^{\updownarrow}\sem f$ are planar tiles.
Since the edge weights in the reduction are polynomial,
the unweighted version in Theorem~\ref{thm:twistedhard} 
follows immediately via Proposition~\ref{pro:weighted}.

\begin{proof}
Membership of the ``$\tcrn(T^{\updownarrow})\leq k$'' 
problem in NP is trivial by a folklore argument;
we may guess the at most $k$ crossings of an optimal drawing,
replace those by new vertices and test planarity of the new tile.
We provide the hardness reduction from Theorem~\ref{thm:anchoredhard}.
Let $(G,A,\sigma)$ be an anchored graph anticipated in
Theorem~\ref{thm:anchoredhard}.
Then $G$ is a disjoint union of two connected components $H_1$ and $H_2$ 
where each of the corresponding anchored subgraphs $H_1,H_2$
is a planar anchored graph.
Let $a=|\sigma|$ and $\sigma'$ be an ordinary (non-cyclic) 
sequence obtained from $\sigma$ by ``opening it'' at any position
such that $\sigma'(1)\in V(H_1)$ and $\sigma'(a)\in V(H_2)$.

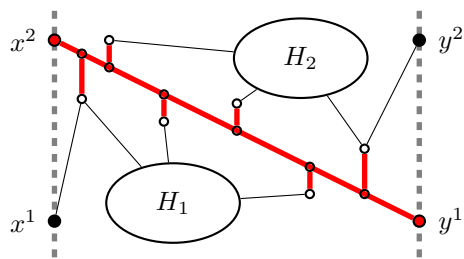
\begin{figure}[t]
\begin{center}\bigskip
\begin{tikzpicture}[scale=1.2]
\normalsize
\tikzstyle{every node}=[draw, thick, shape=circle, minimum size=3pt,inner sep=1.5pt, fill=black]
\tikzstyle{every path}=[color=black]
\draw[line width=2pt, dashed, color=gray] (0,-0.4) -- (0,2.4);
\draw (0,0) node[label=left:$x^1$] (x1) {};
\draw (0,2) node[label=left:$x^2$, fill=red] (x2) {};
\draw[line width=2pt, dashed, color=gray] (4,-0.4) -- (4,2.4);
\draw (4,2) node[label=right:$y^2$] (y2) {};
\draw (4,0) node[label=right:$y^1$, fill=red] (y1) {};
\draw[line width=2pt, color=red] (x2) -- (y1);
\tikzstyle{every node}=[draw, thick, shape=circle, minimum size=2pt,inner sep=1.1pt, fill=none]
\draw (1.3,0.2) node[ellipse, minimum width=50pt, minimum height=30pt] (G1) {$H_1$};
\draw (2.7,1.8) node[ellipse, minimum width=50pt, minimum height=30pt] (G2) {$H_2$};
\draw (2,1) node (d2) {}; \draw (0.3,1.85) node (d5) {};
\draw (2,1.3) node (dd2) {}; \draw (0.3,1.35) node (dd5) {};
\draw (2.8,0.6) node (d1) {}; \draw (3.4,0.3) node (d4) {};
\draw (2.8,0.3) node (dd1) {}; \draw (3.4,0.8) node (dd4) {};
\draw (1.2,1.4) node (d3) {}; \draw (0.6,1.7) node (d6) {};
\draw (1.2,1.1) node (dd3) {}; \draw (0.6,2.0) node (dd6) {};
\draw (dd1) -- (G1) -- (dd3) ; \draw (dd5) -- (G1) ;
\draw (dd2) -- (G2) -- (dd4) ; \draw (dd6) -- (G2) ;
\draw (dd5) -- (x1) ; \draw (dd4) -- (y2) ;
\tikzstyle{every path}=[line width=2pt, color=red]
\draw (d1) -- (dd1) ; \draw (d2) -- (dd2) ;
\draw (d3) -- (dd3) ; \draw (d4) -- (dd4) ;
\draw (d5) -- (dd5) ; \draw (d6) -- (dd6) ;
\end{tikzpicture}
\end{center}
\caption{Constructing a twisted diagonally separated planar tile;
	proof of Corollary~\ref{cor:diagseptile}.}
\label{fig:anchoredtotile}
\end{figure}

Let $Q$ be a path on the vertex set $(x^2,s_1,s_2,\dots,s_a,y^1)$ in this
order and let $x^1,y^2$ be isolated vertices.
We construct a graph $G_0$ from the disjoint union 
$G\cup Q\cup \{x^1,y^2\}$ by adding the following edges:
the edges $\{x^1,\sigma'(1)\}$ and $\{y^2,\sigma'(a)\}$,
and the edges $\{s_i,\sigma'(i)\}$ for $i=1,2,\dots,a$.
All the edges incident with $V(Q)$ get weight
$t=(|E(H_1)|+1)\cdot(|E(H_2)|+1)+1$, while the remaining edges have weight~$1$.
Observe (Figure~\ref{fig:anchoredtotile}) that
$T_0:=\big(G_0,(x^1,x^2),(y^1,y^2)\big)$ is a diagonally separated planar
tile by Definition~\ref{def:diagonalsep}.

We claim that $\acrn(G)\leq k$ if and only if
$\tcrn(T_0^{\updownarrow})\leq k$.
In the forward direction, we take an anchored drawing of $(G,A,\sigma)$
achieving $\acrn(G)$ crossings.
This drawing immediately gives (see also Figure~\ref{fig:diagonalsep} right)
a tile drawing of $T_0^{\updownarrow}$ in which
``thick'' $Q$ and its incident edges, and the vertices $x^1$ and $y^2$,
are all drawn along the boundary of the anchored drawing
without additional crossings.
So, indeed, $\tcrn(T_0^{\updownarrow})\leq \acrn(G)\leq k$.

In the backward direction, we observe that there is a valid tile drawing of
$T_0^{\updownarrow}$ in which the only crossings are between the edges from
$E(H_1)\cup\{\{x^1,\sigma'(1)\}\}$ and the edges from 
$E(H_2)\cup\{\{y^2,\sigma'(a)\}\}$.
Consequently, $\tcrn(T_0^{\updownarrow})\leq
	 (|E(H_1)|+1)\cdot(|E(H_2)|+1)=t-1$.
Assume a tile drawing $D_0$
(recall, $D_0$ is contained in a unit square $\Sigma$ 
with its walls on the left and right sides of~$\Sigma$) of $T_0^{\updownarrow}$\,
with $\tcrn(T_0^{\updownarrow})$ crossings.
By the previous, no ``thick'' edge incident with $Q$ is crossed in~$D_0$.
Since each of the subgraphs $H_1+\{x^1,\sigma'(1)\}$ and
$H_2+\{y^2,\sigma'(a)\}$ is connected,
and both $x^1,y^2$ are positioned to the same side of the ends $x^2,y^1$
of~$Q$ on the boundary of $\Sigma$,
both subgraphs $H_1$ and $H_2$ of $G$ are drawn in the same
region of $\Sigma$ separated by the drawing of~$Q$.
Contracting the uncrossed (``thick'') edges $\{s_i,\sigma'(i)\}$ for
$i=1,2,\dots,a$ hence results in an anchored drawing
of $G$ with at most $\tcrn(T_0^{\updownarrow})$ crossings.
The proof is finished.
%
\end{proof}

\section{Cross-composing}
\label{sec:cross-composing}

We now prove the main result, Theorem~\ref{thm:nokernel}.
By Theorem~\ref{thm:nopolykernel} we know that it is enough to construct an
{\sc or}-cross-composition, that is an algorithm satisfying the requirements
of Definition~\ref{def:crosscomposition}.

\begin{lemma}
\label{lem:tilecr-compose}
Let $\ca L$ be the language of instances $\tuple{T^{\updownarrow},k}$ where
$T$ is a diagonally separated planar tile and $k$ an integer polynomially
bounded in $|T|$,
such that $\tcrn(T^{\updownarrow})\leq k$.
Let an equivalence relation $\sim$ be given as
$\tuple{T_1^{\updownarrow},k_1}\sim\tuple{T_2^{\updownarrow},k_2}$
iff $k_1=k_2$.

Then $\ca L$ admits an {\sc or}-cross-composition, with respect to $\sim$,
into the graph crossing number problem ``$\crg(G)\leq k$'' parameterized by~$k$.
Moreover, this is true even if we restrict $G$ to be an almost-planar graph.
\end{lemma}

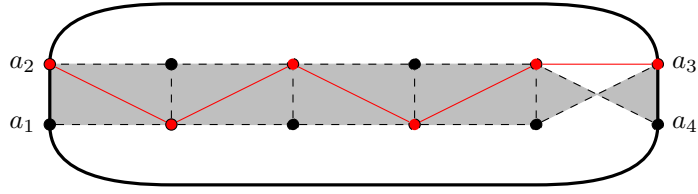
\begin{figure}[t]
\begin{center}\bigskip
\begin{tikzpicture}[scale=0.4]
\normalsize
\tikzstyle{every node}=[draw, shape=circle, minimum size=2.5pt,inner sep=1.5pt, fill=black]
\tikzstyle{every path}=[color=lightgray, fill=lightgray]
\draw (x1) rectangle (16,2) ;
\draw (16,0) -- (20,2) -- (20,0) -- (16,2) -- (16,0) ;
\tikzstyle{every path}=[color=black]
\draw[line width=1.2pt] (0,0) -- (0,2) to
	 [out=90,in=180] (4,4) -- (16,4) to [out=0,in=90] (20,2) --
	 (20,0) to [out=270,in=0] (16,-2) -- (4,-2) to
	 [out=180,in=270] (0,0) ;
\draw (0,0) node[label=left:$a_1$] (x1) {};
\draw (0,2) node[fill=red, label=left:$a_2$] (x2) {};
\draw (4,0) node[fill=red] (x3) {}; \draw (4,2) node (x4) {};
\tikzstyle{every path}=[color=black, dashed]
\draw (x1) -- (x3) -- (x4) -- (x2) ;
\draw (8,0) node (x5) {}; \draw (8,2) node[fill=red] (x6) {};
\draw (x3) -- (x5) -- (x6) -- (x4);
\draw (12,0) node[fill=red] (x7) {}; \draw (12,2) node (x8) {};
\draw (x5) -- (x7) -- (x8) -- (x6);
\draw (16,0) node (x9) {}; \draw (16,2) node[fill=red] (x10) {};
\draw (x7) -- (x9) -- (x10) -- (x8);
\draw (20,2) node[fill=red, label=right:$a_3$] (x11) {};
\draw (20,0) node[label=right:$a_4$] (x12) {};
\draw (x9) -- (x11) -- (x12) -- (x10);
\tikzstyle{every path}=[color=red]
\draw (x2) -- (x3) -- (x6) -- (x7) -- (x10) -- (x11) ;
\end{tikzpicture}
\end{center}
\caption{A sketch of the construction of $G$ in the proof of
	Lemma~\ref{lem:tilecr-compose}; the cycle $C_0$ is in bold and 
	the tiles of $U$ are shaded gray.}
\label{fig:crossintocrossing}
\end{figure}

\begin{proof}
Assume we are given $t$ equivalent instances $\tuple{T_i^{\updownarrow},k}$,
$i=1,2,\dots,t$, of the tile crossing number problem 
$\ca L$; ``$\tcrn(T_i^{\updownarrow})\leq k$''.
Each $T_i$ is a diagonally separated planar tile.
We construct a weighted graph $G$ as follows
(see also Figure~\ref{fig:crossintocrossing} and Lemma~\ref{lem:twistedorcompose}):
\begin{itemize}
\item Let $C_0$ be a cycle on four vertices $a_1,a_2,a_3,a_4$ in this cyclic
order, and all edges of $C_0$ having weight $k+1$.
\item Let $\ca T=(T_1^{\updownarrow},T_2^{\updownarrow}, \ldots, T_t^{\updownarrow})$.
Let $U:=\otimes\ca T$ if $m$ is odd, and $U:=(\otimes\ca T)^{\updownarrow}$ otherwise.
\item $G$ results from the union of $C_0$ and $U$ by identifying,
in the prescribed order, the left wall of $U$ with $(a_1,a_2)$
and the right wall of $U$ with $(a_4,a_3)$.
\end{itemize}

We show that $\crg(G)\leq k$ iff $\tcrn(U)\leq k$.
In the backward direction, any tile drawing of $U$ with $\ell$ crossings
gives a drawing of $G$ with $\ell$ crossings simply by embedding $C_0$
``around'' the tile~$U$.
Conversely, assume a drawing $D$ of $G$ with $\ell\leq k$ crossings,
and observe that no edge of $C_0$ (weighted $k+1$) is crossed in $D$.
Since $G\sem C_0$ is connected, it is drawn in one of the two faces of $C_0$
and this clearly gives a tile drawing of~$U$ with $\ell$ crossings.

Now, by Lemma~\ref{lem:twistedorcompose},
$\tcrn(U)\leq k$ iff there exists $i\in\{1,\dots,t\}$ such that
$\tcrn(T_i^{\updownarrow})\leq k$, as required by
Definition~\ref{def:crosscomposition}.
The construction of $G$ is easily finished in polynomial time,
and since the edge weights $k+1$ in $G$ are polynomially bounded,
there is a polynomial reduction to an unweighted crossing number instance by
Proposition~\ref{pro:weighted}.
It remains to verify that $G$ is almost-planar.
Let $e_1$ be the unique edge of $T_1$ incident with $a_1$ in~$G$.
Then $\tcrn(T_1^{\updownarrow}\sem e_1)=\tcrn(T_1\sem e_1)=0$
and hence~$\crg(G\sem e_1)=0$.
\end{proof}

Theorem~\ref{thm:nokernel} follows from Corollary~\ref{cor:diagseptile}
and Lemma~\ref{lem:tilecr-compose} via Theorem~\ref{thm:nopolykernel}
(note that $\sim$ trivially is a polynomial equivalence).

\section{Conclusion}
\label{sec:conclusion}

We have proved that the graph crossing number problem parameterized by
the number of crossings, which is known to be fixed parameter tractable,
is highly unlikely to admit a polynomal kernelization.
The complexity of the crossing number problem has been commonly studied 
under various additional restrictions on the input graph.
Our negative result extends even to the instances in which the input graph
$G$ is one edge away from planarity (i.e., almost-planar~$G$).

On the other hand, the ordinary crossing number problem remains NP-hard 
for cubic graphs and for the so-called minor crossing 
number~\cite{DBLP:journals/jct/Hlineny06a}, and for graphs with a
prescribed edge rotation system~\cite{DBLP:journals/algorithmica/PelsmajerSS11}.
For a drawing of a graph, the {\em rotation} of a vertex 
is the clockwise order of its incident edges (in a local neighbourhood).
A {\em rotation system} is the list of rotations of every vertex.
As proved in~\cite{DBLP:journals/algorithmica/PelsmajerSS11}, there is a
polynomial equivalence between the problems of computing the crossing number
of cubic graphs and that of computing the crossing number under prescribed
rotation systems. 

The construction we use to show hardness in the paper,
produces instances which are ``very far'' from having small vertex degrees or
a fixed rotation system, and there does not seem to be any easy modification for that.
Nevertheless, we have an indication that the following strengthening might
also be true:

\begin{conjecture}
Let $G$ be a graph with a given rotation system.
Let $k\geq1$ be an integer.
The problem of whether there is a drawing of $G$ respecting the prescribed
rotation system and having at most $k$ crossings, parameterized by~$k$,
does not admit a polynomial kernel unless \nokernelhypo.

Consequently, the crossing number problem $\crg(G)\leq k$ restricted 
to cubic graphs $G$, and the analogous minor crossing number problem,
do not admit a polynomial kernel w.r.t.~$k$ unless \nokernelhypo.
\end{conjecture}

\subsection*{Acknowledgements}
This research was supported by the Czech Science Foundation project No.~14-03501S.
We would also like to thank the anonymous referees for helpful comments,
and the organizers of the Workshop on Kernelization
2015, at the University of Bergen, Norway, where the idea of this paper was
born.

\bibliography{phcross,gtbib}

\end{document}